\documentclass[11pt,a4paper]{article}

\usepackage[utf8]{inputenc}
\usepackage{mathtools}
\usepackage{amsfonts}
\usepackage{amssymb}
\usepackage{amsthm}
\usepackage{tikz}
\usetikzlibrary{arrows,automata}

\newtheorem{theorem}{Theorem}
\newtheorem{definition}[theorem]{Definition}
\newtheorem{lemma}[theorem]{Lemma}

\newtheorem{corollary}[theorem]{Corollary}
\newtheorem{example}[theorem]{Example}

\newcommand{\tuple}[1]{\langle #1 \rangle}
\newcommand{\limrun}[0]{\infty}
\newcommand{\emptyword}[0]{\lambda}
\newcommand{\trans}[1]{\mathchoice{\xrightarrow{#1}}{\xrightarrow{\smash{\lower1pt\hbox{$\scriptstyle #1$}}}}{\text{Error}}{\text{Error}}}
\newcommand{\occ}[2]{|#1|_{#2}}
\newcommand{\prefsel}[0]{\upharpoonright}

\DeclareUnicodeCharacter{B7}{\cdot}             
\DeclareUnicodeCharacter{D7}{\times}            
\DeclareUnicodeCharacter{394}{\Delta}           
\DeclareUnicodeCharacter{3B2}{\beta}            
\DeclareUnicodeCharacter{3B4}{\delta}           
\DeclareUnicodeCharacter{3B5}{\varepsilon}      
\DeclareUnicodeCharacter{3B7}{\eta}             
\DeclareUnicodeCharacter{3B8}{\theta}           
\DeclareUnicodeCharacter{3B9}{\iota}            
\DeclareUnicodeCharacter{3BA}{\kappa}           
\DeclareUnicodeCharacter{3BB}{\lambda}          
\DeclareUnicodeCharacter{3BC}{\mu}              
\DeclareUnicodeCharacter{3BE}{\xi}              
\DeclareUnicodeCharacter{3C0}{\pi}              
\DeclareUnicodeCharacter{3C1}{\rho}             
\DeclareUnicodeCharacter{3C3}{\sigma}           
\DeclareUnicodeCharacter{2026}{\ldots}          
\DeclareUnicodeCharacter{2115}{\mathbb{N}}      
\DeclareUnicodeCharacter{2133}{\mathcal{M}}     
\DeclareUnicodeCharacter{2192}{\rightarrow}     
\DeclareUnicodeCharacter{2203}{\exists}         
\DeclareUnicodeCharacter{2208}{\in}             
\DeclareUnicodeCharacter{2209}{\notin}          
\DeclareUnicodeCharacter{221E}{\infty}          
\DeclareUnicodeCharacter{2260}{\neq}            
\DeclareUnicodeCharacter{2282}{\subset}         
\DeclareUnicodeCharacter{2286}{\subseteq}       
\DeclareUnicodeCharacter{2291}{\sqsubseteq}     
\DeclareUnicodeCharacter{22C5}{\cdot}           
\DeclareUnicodeCharacter{22EF}{\cdots}          
\DeclareUnicodeCharacter{2A7D}{\leqslant}       
\DeclareUnicodeCharacter{2A7E}{\geqslant}       
\DeclareUnicodeCharacter{1D49C}{\mathcal{A}}    
\DeclareUnicodeCharacter{1D4AE}{\mathcal{S}}    
\DeclareUnicodeCharacter{1D4AF}{\mathcal{T}}    

\begin{document}

\title{A direct proof of Agafonov's theorem \\
       and an extension to shifts of finite type}
\author{Olivier Carton}

\date{\today}
\maketitle

\begin{abstract}
  We provide a direct proof of Agafonov's theorem which states that finite
  state selection preserves normality.  We also extends this result to the
  more general setting of shifts of finite type by defining selections
  which are compatible with the shift.  A slightly more general statement is
  obtained as we show that any Markov measure is preserved by finite state
  compatible selection.
\end{abstract}


\section{Introduction}

Normality was introduced by Borel in~\cite{Borel09} more than one hundred
years ago to formalize the most basic form of randomness for real
numbers. A number is normal to a given integer base if its expansion in
that base is such that all blocks of digits of the same length occur in it
with the same limiting frequency.

Although normality is a purely combinatorial property, it has close links
with finite state machines.  A fundamental theorem relates normality and
finite automata: an infinite sequence is normal to a given alphabet if and
only if it cannot be compressed by lossless finite transducers. These are
deterministic finite automata with injective input-output behavior. This
result was first obtained by joining a theorem by Schnorr and Stimm
\cite{SchnorrStimm71} with a theorem by Dai, Lathrop, Lutz and Mayordomo
\cite{Dai04}. Becher and Heiber gave a direct proof
in~\cite{BecherHeiber13}.  Another astonishing result is Agafonov's theorem
stating that selecting symbols in a normal sequence using a finite state
machine preserves normality \cite{Agafonov68}.  Agafonov’s publication
\cite{Agafonov68} does not really include the proof but O’Connor
\cite{OConnor88} provided it using predictors defined from finite automata,
and Broglio and Liardet \cite{BroglioLiardet92} generalized it to arbitrary
alphabets.  Later Becher and Heiber gave another proof based of the
characterization of normality by non-compressibility by lossless finite
transducers \cite{BecherHeiber13}.  In this paper, we provide a direct
proof of Agafonov's theorem.  The proof is almost elementary but it still
relies on Markov chains arguments.

The notion of normality has been extended to broader contexts like the one
of dynamical systems and especially shifts of finite type
\cite{Madritsch18}.  When sofic shifts are irreducible and aperiodic, they
have a measure of maximal entropy and a sequence is then said to be normal
if the frequency of each block equals its measure.  This extension to
shifts meets the original aim of normality to study expansions of numbers
in bases when the shift arises from a numerical systems like the $β$-shifts
coming from the numeration in a non-integer base~$β$.  Normality can be
again interpreted as the good distribution of blocks of digits in the
expansion of a number in a base~$β$.  In this paper, we extend Agafonov's
theorem to the setting of shift of finite type.  More precisely, we show
that genericity for Markovian measure is preserved by selection with finite
state state machines if the machines satisfy some compatibility condition
with the measure.  This result includes the case of shifts of finite type
as their Parry measure is Markovian.

The paper is organized as follows.  Section~\ref{sec:prelim} is devoted to
notation and main definitions.  The link between selection and special
finite-state machines called selectors is given in
Section~\ref{sec:selection}.  Agafanov's theorem is stated and proved in
Section~\ref{sec:agafanov}.  The extension of the theorem to Markovian
measures is given in Section~\ref{sec:markov}.  Note that the proof given
that section subsumes the one given in the previous one.  We keep both
proofs since we think that the one in Section~\ref{sec:agafanov} is a nice
preparation for the reader to the one in Section~\ref{sec:markov}.

\section{Preliminaries} \label{sec:prelim}

\subsection{Sequences, shifts and selection}

We write $ℕ$ for the set of all non-negative integers.  Let $A$ be a finite
set.  We let $A^*$ and $A^ℕ$ respectively denote the sets of all finite and
infinite sequences over the alphabet~$A$.  Similarly $A^k$ stands for the
set of sequences of length~$k$.  Finite sequence are also called
\emph{words}. The empty word is denoted by~$\emptyword$ and the length of a
word~$w$ is denoted $|w|$.  The positions in finite and infinite words are
numbered starting from~$1$.  For a word~$w$ and positions
$1 ⩽ i ⩽ j ⩽ |w|$, we let $w[i]$ and $w[i{:}j]$ denote respectively the
symbol~$a_i$ at position~$i$ and the word $a_ia_{i+1}⋯ a_j$ from
position $i$ to position~$j$.  A word of the form $w[i{:}j]$ is called a
\emph{block} of~$w$.  A word~$u$ is a \emph{prefix} (respectively
\emph{suffix}) of a word~$w$, denoted $u ⊑ w$, if $w = uv$ (respectively
$w = vu$) for some word~$v$.

For any finite set $S$ we denote its cardinality with $\#S$.  We write
$\log$ for the base~$2$ logarithm.

In this article we are going to work on shift spaces, in particular
shifts of finite type (SFT).  Let $A$ be a given alphabet.  The
\emph{full shift} is the set $A^ℕ$ of all (one-sided) infinite sequences
$(x_n)_{n⩾1}$ of symbols in~$A$.  The shift~$σ$ is the function from~$A^ℕ$
to~$A^ℕ$ which maps each sequence $(x_n)_{n⩾1}$ to the sequence
$(x_{n+1})_{n⩾1}$ obtained by removing the first symbol.

A \emph{shift space} of~$A^ℕ$ or simply a \emph{shift} is a subset~$X$
of~$A^ℕ$ which is closed for the product topology and invariant under the
shift operator, that is $σ(X) = X$.  Let $F ⊂ A^*$ be a set of finite words
called \emph{forbidden blocks}.  The shift~$X_F$ is the subset of~$A^ℕ$
made of sequences without any occurrences of blocks in~$F$.  More formally,
it is the set
\begin{displaymath}
  X_F = \{ x : x[m{:}n] ∉ F \text{ for each } 1 ⩽ m ⩽ n\}.
\end{displaymath}

It is well known that a shift~$X$ is characterized by its forbidden blocks,
that is $X = X_F$ for some set $F ⊂ A^*$.  The shift~$X$ is said to be of
\emph{finite type} if $X = X_F$ for some finite set~$F$ of forbidden blocks
\cite[Def.~2.1.1]{LindMarcus92}.  Up to a change of alphabet, any shift
space of finite type is the same as a shift space~$X_F$ where any forbidden
block has length~$2$, that is $F ⊂ A^2$.

For simplicity, we always assume that each forbidden block has length~$2$.
In that case, the set~$F$ is given by an $A × A$-matrix
$P = (p_{ab})_{a,b ∈ A}$ where $p_{ab} = 0$ if $ab ∈ F$ and $p_{ab} > 0$
otherwise and we write $X = X_P$.  The shift~$X$ is called
\emph{irreducible} if the graph induced by the matrix~$P$ is strongly
connected, that is, for each symbols $a,b ∈ A$, there exists an integer~$n$
(depending on $a$ and~$b$) such that $P^n_{ab} > 0$.  The shift~$X$ is
called \emph{irreducible} and \emph{aperiodic} if there exists an
integer~$n$ such that $P^n_{ab} > 0$ for each symbols $a,b ∈ A$.

\begin{example}[Golden mean shift]
  The \emph{golden mean shift} is the shift space $X_F \subset \{0,1\}^ℕ$
  where the set of forbidden blocks is $F = \{ 11 \}$.  It is made of all
  sequences over $\{0,1\}$ with no two consecutive~$1$.  This subshift is
  also equal to $X_M$ where the matrix~$M$ is given by
  $M = \left(\begin{smallmatrix} 1&1 \\ 1&0
    \end{smallmatrix}\right)$.
\end{example}

Let $x = a_1a_2a_3 ⋯ $ be a sequence over the alphabet~$A$.  Let $L ⊆ A^*$
be a set of finite words over~$A$. The word obtained by \emph{oblivious
  prefix selection} of~$x$ by~$L$ is $x \prefsel L = a_{i_1}a_{i_2}a_{i_3}
⋯$ where $i_1,i_2,i_3,…$ is the enumeration in increasing order of all the
integers~$i$ such that the prefix $a_1 a_2 ⋯ a_{i-1}$ belongs to~$L$.  This
selection rule is called \emph{oblivious} because the symbol~$a_i$ is not
included in the considered prefix.  If $L = A^*1$ is the set of words
ending with a~$1$, the sequence~$x \prefsel L$ is made of all symbols
of~$x$ occurring after a~$1$ in the same order as they occur in~$x$.

\subsection{Measures and genericity}

A \emph{probability measure on $A^*$} is a function
$μ : A^* \rightarrow [0,1]$ such that $μ(\emptyword) = 1$ and
\begin{displaymath}
  \sum_{a ∈ A}{μ(wa)} = μ(w)
\end{displaymath}
holds for each word $w ∈ A^*$.  The simplest example of a probability
measure is a \emph{Bernoulli measure}.  It is a monoid morphism from~$A^*$
to~$[0,1]$ (endowed with multiplication) such that
$\sum_{a ∈ A}{μ(a)} = 1$.  Among the Bernoulli measures is the
\emph{uniform measure} which maps each word $w ∈ A^*$ to $(\#A)^{-|w|}$.
In particular, each symbol~$a$ is mapped to $μ(a) = 1/\#A$.

By the Carathéodory extension theorem, a measure~$μ$ on~$A^*$ can be
uniquely extended to a probability measure~$\hat{μ}$ on~$A^ℕ$ such that
$\hat{μ}(wA^ℕ) = μ(w)$ holds for each word $w ∈ A^*$.  In the rest of the
paper, we use the same symbol for $μ$ and~$\hat{μ}$.  A probability
measure~$μ$ is said to be \emph{(shift) invariant} if the equality
\begin{displaymath}
  \sum_{a ∈ A}{μ(aw)} = μ(w)
\end{displaymath}
holds for each word $w ∈ A^*$.

We now recall the definition of Markov measures.  For a stochastic
matrix~$P$ and a stationary distribution~$π$, that is a raw vector such
that $π P = π$, the Markov measure $μ_{π,P}$ is the invariant measure
defined by the following formula \cite[Lemma~6.2.1]{Kitchens98}.
\begin{displaymath}
  μ_{π,P} (a_1 a_2 ⋯ a_k) = π_{a_1} P_{a_1a_2} ⋯ P_{a_{k-1} a_k}
\end{displaymath}

A measure~$μ$ is \emph{compatible} with a shift~$X_F$ if it only puts
weight on blocks of~$X$, that is, $μ(w) > 0$ implies $w ∉ F$ for each
word~$w$.  For a shift of finite type, there is a unique compatible measure
with maximal entropy \cite[Thm.~6.2.20]{Kitchens98}. This measure is called
the \emph{Parry measure} and it is a Markov measure.  This measure can be
explicitly given as follows.  The Parry measure of a SFT~$X_M$ is the
Markov measure given by the stochastic matrix $P = (P_{i,j})$ where
$P_{i,j} = M_{i,j} r_j/θ r_i$ and the stationary probability
distribution~$π$ defined by $π_i = l_ir_i$, where $θ$ is the Perron
eigenvalue of the matrix~$M$ and the vectors $l$ and~$r$ are respectively
the left and right eigenvectors of~$M$ for~$θ$ normalized so that
$\sum_{i=1}^k l_ir_i = 1$.

\begin{example}[Parry measure of the golden mean shift]
  
  Consider again the golden mean shift~$X$.  Its Parry measure is the
  Markov measure $μ_{π,P}$ where $π$ is the distribution $π = (λ^2/(1+λ^2),
  1/(1+λ^2))$ and $P$ is the stochastic matrix
  $P = \left(\begin{smallmatrix} 1/λ & 1/λ^2 \\
    1 & 0 \end{smallmatrix}\right)$ where $λ$ is the golden mean.
\end{example}

Conversely, the \emph{support} of an invariant measure~$μ$ is the shift
$X_μ = X_F$ where $F$ is the set of words of measure zero, that is
$F = \{ w : μ(w) = 0 \}$.  If $μ$ is the Markovian measure~$μ_{π,P}$, then
its support~$X_μ$ is a shift of finite type because it is equal to the
shift~$X_P$ given by matrix~$P$.

We recall here the notion of normality and the notion of genericity.  We
start with the notation for the number of occurrences of a given word~$u$
within another word~$w$.  For two words $u$ and~$w$, the number
$\occ{w}{u}$ of \emph{occurrences} of~$u$ in~$w$ is given by $\occ{w}{u} =
\#\{ i : w[i{:}i+|u|-1] = u \}$.  Borel's definition~\cite{Borel09} of
normality for a sequence $x ∈ A^ℕ$ is that $x$ is \emph{normal} if for each
finite word $w ∈ A^*$
\begin{displaymath}
  \lim_{n → ∞} \frac{\occ{x[1{:}n]}{w}}{n} = (\#A)^{-|w|}
\end{displaymath}
A sequence~$x$ is called \emph{generic} for a measure~$μ$ (or merely
\emph{$μ$-generic}) if for each word $w ∈ A^*$
\begin{displaymath}
  \lim_{n → ∞} \frac{\occ{x[1{:}n]}{w}}{n} = μ(w)
\end{displaymath}
Normality is then the special case of genericity when the measure~$μ$ is
the uniform measure.  There are another definitions of normality and
genericity taking into account only some occurrences, called aligned
occurrences, of each word~$w$.  More precisely, the sequence~$x$ is
factorized $x = w_1w_2w_3⋯$ where $|w_i| = |w|$ for each $i ⩾ 1$ and
it is required that the quotient $\#\{ i ⩽ n : w_i = w\}/n$ converges to
$μ(w)$ when $n$ goes to infinity for each word~$w$.  It is shown
in~\cite{AlvarezCarton19} that the two notions coincide as long as the
measure~$μ$ is Markovian.



\section{Finite-state selection} \label{sec:selection}

In this section, we introduce the automata with output also known as
\emph{transducers} which are used to select symbols from a sequence.  We
consider \emph{deterministic transducers} computing functions from
sequences in a shift~$X$ to sequences in a shift~$Y$, that is, for a given
input sequence $x ∈ X$, there is at most one output sequence $y ∈ Y$.  We
focus on transducer that operate in real-time, that is, they process
exactly one input alphabet symbol per transition.  We start with the
definition of a transducer.

\begin{definition} 
  An \emph{input deterministic transducer} $𝒯$ is a tuple
  $\tuple{Q,A,B,δ,I,F}$, where
  \begin{itemize} \itemsep0cm
  \item $Q$ is a finite set of \emph{states},
  \item $A$ and $B$ are the input and output alphabets, respectively,
  \item $δ: Q × A → B^* × Q$ is the
    \emph{transition} function,
  \item $I \subseteq Q$ and $F \subseteq Q$ are the sets of \emph{initial}
    and \emph{final} states, respectively.
  \end{itemize}
\end{definition}
Input deterministic transducers are also called sequential in the
litterature \cite{Sakarovitch09}.  The relation $δ(p,a) = (w,q)$ is written
$p \trans{a|v} q$ and the tuple $\tuple{p,a,w,a}$ is then called a
\emph{transition} of the transducer.  A finite (respectively infinite)
\emph{run} is a finite (respectively infinite) sequence of consecutive
transitions,
\begin{displaymath}
  q_0 \trans{a_1|v_1} q_1 \trans{a_2|v_2} q_2
  ⋯ q_{n-1} \trans{a_n|v_n} q_n.
\end{displaymath}
Its \emph{input and output labels} are respectively $a_1⋯ a_n$ and
$v_1 ⋯ v_n$.  A finite run is written $q_0 \trans{a_1⋯ a_n|v_1 ⋯ v_n} q_n$.
An infinite run is written $q_0 \trans{a_1a_2a_3⋯|v_1v_2v_3⋯} \limrun$.  An
infinite run is accepting if its first state~$q_0$ is initial.  Note that
there is no accepting condition.  This is due to the fact that we always
assume that the domain is a closed subset of~$A^ℕ$.  Since transducers are
assumed to be input deterministic there is at most one run with input
label~$x$ for each $x$ in~$A^ℕ$.  If the output label is the infinite
sequence~$y$, we write $y = 𝒯(x)$.  By a slight abuse of notation, we write
$𝒯(x[m{:}n])$ for the output of~$𝒯$ along that run while reading the block
$x[m{:}n]$ of~$x$.  We always asumme that all transducers are \emph{trim}:
each state occurs in at least one accepting run.  Since transducers are
input deterministic, the stating state and the input label determine the
run and the ending state.  For a state~$p$ and a word~$u$, we let $p * u$
and $p ⋅ u$ denote respectively the run $p \trans{u|v} q$ and its ending
state~$q$.

A \emph{selector} is a deterministic transducer such that each of its
transitions has one of the types $p \trans{a|a} q$ (type~I),
$p \trans{a|\emptyword} q$ (type~II) for a symbol $a ∈ A$. In a selector,
the output of a transition is either the symbol read by the transition
(type~I) or the empty word (type~II).  Therefore, it can be always assumed
that the output alphabet~$B$ is the same as the input alphabet~$A$.  It
follows that for each run $p \trans{u|v} q$, the output label~$v$ is a
subword, that is a subsequence, of the input label~$u$.

\begin{figure}[htbp]
  \begin{center}
  \begin{tikzpicture}[->,>=stealth',initial text=,semithick,auto,inner sep=1.5pt]
    \tikzstyle{every state}=[minimum size=0.4]
    \node[state,initial left] (q0) at (1.25,2) {$q_0$};
    \node[state]  (q1) at (0,0) {$q_1$};
    \node[state]  (q2) at (2.5,0) {$q_2$};
    \path (q0) edge[out=30,in=-30,loop] node {$0|0$} ();
    \path (q0) edge[bend left=20] node {$1|1$} (q1);
    \path (q1) edge[bend left=20] node {$0|0$} (q0);
    \path (q1) edge[swap] node {$1|\emptyword$} (q2);
    \path (q2) edge[out=30,in=-30,loop] node {$0|\emptyword$} ();
    \path (q2) edge[swap] node {$1|1$} (q0);
  \end{tikzpicture}
  \end{center}
  \caption{A selector}
  \label{fig:select0}
\end{figure}
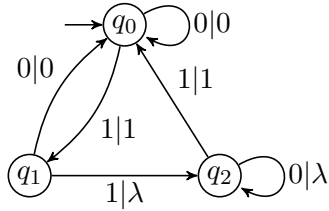

A selector is \emph{oblivious} if all transitions starting from a given
state have the same type.  The selector pictured in
Figure~\ref{fig:select0} is not oblivious but the one pictured in
Figure~\ref{fig:select1} is oblivious.  The terminology is justified by the
following relation between oblivious prefix selection and selectors.  If
$L ⊆ A^*$ is a rational set, the oblivious prefix selection by~$L$ can be
performed by an oblivious selector.  There is indeed an oblivious
selector~$𝒮$ such that for each input word~$x$, the output $𝒮(x)$ is the
result $x \prefsel L$ of the selection by~$L$.  This selector~$𝒮$ can be
obtained from any deterministic automaton~$𝒜$ accepting~$L$.  Replacing
each transition $p \trans{a} q$ of~$𝒜$ by either $p \trans{a|a} q$ if the
state~$p$ is accepting or by $p \trans{a|\emptyword} q$ otherwise yields
the selector~$𝒮$.  It can be easily verified that the obtained transducer
is an oblivious selector performing the oblivious prefix selection by~$L$.
Conversely, each oblivious selector performs the oblivious prefix selection
by~$K$ where $K$ is the set of words being the input label of a run from
the initial state to a state~$q$ such that transitions starting from~$q$
have type~I.

\begin{figure}[htbp]
  \begin{center}
  \begin{tikzpicture}[->,>=stealth',initial text=,semithick,auto,inner sep=1.5pt]
    \tikzstyle{every state}=[minimum size=0.4]
    \node[state,initial above] (q0) at (0,0) {$q_0$};
    \node[state]  (q1) at (2,0) {$q_1$};
    \path (q0) edge[out=210,in=150,loop] node {$0|\emptyword$} ();
    \path (q0) edge[bend left=20] node {$1|\emptyword$} (q1);
    \path (q1) edge[bend left=20] node {$0|0$} (q0);
    \path (q1) edge[out=30,in=-30,loop] node {$1|1$} ();
  \end{tikzpicture}
  \end{center}
  \caption{Oblivious selector that selects symbols following a~$1$}
  \label{fig:select1}
\end{figure}
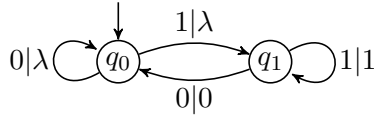

The transducer pictured in Figure~\ref{fig:select1} is an oblivious
selector that selects symbols occurring after a~$1$.  It performs the
oblivious prefix selection by~$L$ where $L$ is the set $A^*1$ of words
ending with a~$1$.

Some reasoning about transducers only involve the input labels of
transitions and ignore the output labels.  We call \emph{automaton} a
transducer where output labels of transitions are removed.  This means that
the transition function~$δ$ is then a function from $Q × A$ to~$Q$ where
$Q$ is the state set and $A$ the alphabet.

\section{Preservation of  normality} \label{sec:agafanov}

In this section, we consider normality in the full shift.  We give an
alternative proof of Agafonov's result \cite{Agafonov68} that finite state
selection preserves normality.  This means that if the sequence~$x$ is
normal and $L$ is a regular set of finite words, then the
sequence~$x \prefsel L$ is still normal.  Since it has been remarked that
selecting by a regular set is the same as using an oblivious selector, the
result means that if $𝒮$ is an oblivious selector and $x$ is normal, then
$𝒮(x)$ is also normal.

The following theorem of Agafonov states that oblivious prefix selection by
a regular set preserves normality.
\begin{theorem}[Agafonov \cite{Agafonov68}] \label{thm:selection}
  If $x$ is normal and $L$ is regular, then $x \prefsel L$ is still normal.
\end{theorem}

The strategy of the proof is the following.  We consider an oblivious
selector~$𝒮$ performing selection by~$L$.  This means that if $x$ and $y$
are the input and output label of successful run, then $y = x \prefsel L$.
We show then that if the input label~$x$ a normal sequence, then the output
of the run of~$𝒮$ is also normal.  We fix a state~$p$ of~$𝒮$ and an
integer~$\ell$.  We show that for $k$ great enough, the number of runs
starting from~$p$ and outputting less than $\ell$ symbols is negligible.
Then we show that for each words $w$ and~$w'$ of length~$\ell$, the number
of runs outputting $w$ and $w'$ are almost the same.  Finally, we show that
all these runs of lengths~$k$ starting from~$p$ have the same frequency in
a run whose input is a normal word.

The following lemma shows that the number of runs starting from a 
state~$p$ and outputting a fixed word~$w$ is not too large.
\begin{lemma} \label{lem:upper}
  Let $𝒮$ be an oblivious selector.  For each state~$p$
  of~$𝒮$, each integer~$n ⩾ 0$, and word~$w ∈ A^*$ such that
  $|w| ⩽ n$, there are at most $(\#A)^{n-|w|}$ runs $p \trans{u|v} q$ of
  length~$n$ such that $w$ is a prefix of the output label~$v$.
\end{lemma}
\begin{proof}
  The proof is carried out by induction on the integer~$n$.  If $n = 0$,
  the only possible word is the empty word~$\emptyword$.  Since there is
  only one run of length~$0$, the inequality is satisfied.  We now
  suppose that $n ⩾ 1$. Since the selector is oblivious, all transitions
  starting from state~$p$ have the same type, either type~I or type~II.

  We first suppose that all transitions starting from state~$p$ have type~I.
  Let us write $w = aw'$ where $a$ is a symbol and $w'$ a word.  Consider
  the transition $p \trans{a|a} q$.  All runs starting from~$p$ such that
  $w$ is a prefix of the output label must use this transition as a first
  transition.  Applying the induction hypothesis to $q$, $n-1$ and~$w'$
  gives the result.

  We now suppose that all transitions starting from state~$p$ have type~II,
  that is, have the form $p \trans{a|\emptyword} q_a$ for each symbol~$a$.
  This implies that all runs of length~$n$ starting from~$p$ have an output
  label of length at most $n-1$.  Therefore, if $|w| = n$, there is no run
  such that $w$ is prefix of its output label and the inequality is
  trivially satisfied.  If $|w| ⩽ n-1$, applying the induction hypothesis
  to each~$q_a$, $n-1$, and~$w$ gives that the number of runs starting
  from~$q_a$ such that $w$ is a prefix of their output label is at most
  $(\#A)^{n-1-|w|}$.  Summing up all these inequalities for all~$q_a$ gives
  the required inequality for~$p$.
\end{proof}

Some of the bounds are obtained using the ergodic theorem for Markov chains
\cite[Thm~4.1]{Bremaud08}.  For that purpose, we associate a Markov chain
$ℳ$ to each strongly connected automaton~$𝒜$.  For
simplicity, we assume that the state set~$Q$ of~$𝒜$ is the set
$\{1, …,\#Q\}$.  The state set of the Markov chain is the same set
$\{1, …,\#Q\}$.  The transition matrix of the Markov chain is the
matrix $P = (p_{i,j})_{1⩽ i,j ⩽ \#Q}$ where each entry~$p_{i,j}$ is
equal to $\#\{ a: i \trans{a} j\}/\#A$.  Note that
$\#\{ a: i \trans{a} j\}$ is the number of transitions from~$i$ to~$j$.
Since the automaton is assumed to be deterministic and complete, the
matrix~$P$ is stochastic.  If the automaton~$𝒜$ is strongly
connected, the Markov chain is irreducible and it has therefore a unique
stationary distribution~$π$ such that $π P = π$.  The vector~$π$
is called the \emph{distribution} of~$𝒜$.

By a slight abuse of notation, we let $\occ{p * w}{q}$ denote the number of
occurrences of the state~$q$ in the finite run $p * w$. The idea of the
following lemma is borrowed from~\cite{SchnorrStimm71}.

\begin{lemma} \label{lem:ergodic}
  Let $𝒜$ be a strongly connected deterministic and complete
  automaton and let $π$ be its distribution.  For each real numbers
  $ε,δ > 0$, there exists an integer~$N$ such that for each
  integer $n > N$
  \begin{displaymath}
    \#\left\{
      w ∈ A^n : ∃ p,q ∈ Q \;\;\bigl|\occ{p * w}{q}/n - π_q\bigr| > δ
    \right\} < ε(\#A)^n
  \end{displaymath}
\end{lemma}
\begin{proof}
  The proof is a mere application of the ergodic theorem for
  Markov chains \cite[Thm~4.1]{Bremaud08}.
\end{proof}

The following corollary is also borrowed from~\cite{SchnorrStimm71}.
\begin{corollary} \label{cor:freqstate}
  Let $𝒜$ be a deterministic and strongly connected automaton
  and let $π$ its distribution.  Let $ρ$ the run of~$𝒜$
  on a normal sequence~$x$.  Then for each state~$q$
  \begin{displaymath}
    \lim_{n → ∞}{\frac{\occ{ρ[1{:}n]}{q}}{n}} = π_q.
  \end{displaymath}
  where $ρ[1{:}n]$ is the finite run made of the first $n$ transitions
  of~$ρ$
\end{corollary}
\begin{proof}
  Since $\sum_{q ∈ Q}{π_q} = 1$, it suffices to prove that
  $\limsup_{n → ∞}{\occ{ρ[1{:}n]}{q}/n}  ⩾ π_q$ holds for each
  state~$q$.

  Let $ε>0$ be a positive real number.  Applying
  Lemma~\ref{lem:ergodic} with $δ = ε$ provides an
  integer~$k$ such that
  \begin{displaymath}
    B = \left\{
        w ∈ A^k :
        ∃ p \;\;\bigl|\occ{p * w}{q}/k - π_q\bigr| > ε
    \right\}
  \end{displaymath}
  has cardinality at most $ε(\#A)^n$.
  The run~$ρ$ is then factorized
  \begin{displaymath}
    ρ = p_0 \trans{w_0} p_1 \trans{w_1} p_2  \trans{w_2} p_3 ⋯
         = (q_0 * w_0)(q_1 * w_1)(q_2 * w_2)⋯
  \end{displaymath}
  where each word~$w_i$ is of length~$k$ and $x = w_0w_1w_2⋯$.  Since $x$
  is normal, there is, by Theorem~4 in~\cite{AlvarezCarton19}, an
  integer~$N$ such that for each $n > N$ the cardinality of the set
  $\{ i < n : w_i = w \}$ is greater than $(1-ε)n/(\#A)^{k}$ for each word
  $w ∈ A^k$.

  \begin{align*}
    \limsup_{n → ∞}{\frac{\occ{ρ[1{:}n]}{q}}{n}} & =
           \lim_{n → ∞}{\frac{\occ{ρ[1{:}nk]}{q}}{nk}} \\
    & = \frac{1}{nk}\sum_{i = 0}^{n-1}{\occ{q_i * w_i}{q}} \\
    & ⩾ \frac{1}{nk}\sum_{w ∈ A^k}{\#\{ i < n : w_i = w \}
        × \min_{p∈ Q}\occ{p * w}{q}} \\
    & ⩾ \frac{1}{nk}\sum_{w ∈ A^k \setminus B}
         {((1-ε)n/(\#A)^{k})(k(π_q-ε))} \\
    & = (1-ε)^2(π_q-ε)
  \end{align*}
  Since this inequality holds for each real number $ε>0$, we have
  proved that $\limsup_{n → ∞}{\occ{ρ[1{:}n]}{q}/n} ⩾ π_q$.
\end{proof}

Using the terminology of Markov chains, a strongly connected component
(SCC) of an automaton is called \emph{recurrent} if it cannot be left.
This means that there is no transition $p \trans{a} q$ where $p$ is in that
component and $q$ is not.  The following lemma is Satz~2.5
in~\cite{SchnorrStimm71}.
\begin{lemma} \label{lem:recurscc}
  Let $𝒜$ be an automaton and let $ρ$ be a run of~$𝒜$ on a normal input
  sequence.  The run~$ρ$ reaches a recurrent SCC of~$𝒜$.
\end{lemma}
The hypothesis that the input sequence is normal is stronger than what is
required.  It suffices that each block has infinitely many occurrences in
the sequence.

\begin{lemma} \label{lem:equirun}
  Let $𝒮$ be a strongly connected selector.  For each integer~$k$
  and each real number $ε > 0$, there exists an integer~$N$ such
  that for each integer~$n > N$, each state~$p$ and each word~$w$ of
  length~$k$, the number of runs $p \trans{u|v} q$ of length~$n$ such that
  $w$ is a prefix of the output label~$v$ is between
  $(1-ε)(\#A)^{n-|w|}$ and $(\#A)^{n-|w|}$.
\end{lemma}
\begin{proof}
  Let $p$ be any state. The upper bound $(\#A)^{n-|w|}$ has been already
  proved in Lemma~\ref{lem:upper}.  It remains to prove the lower bound.

  Let fix a state~$q$ such that the transitions starting from~$q$ are of
  type~I.  If no such state exists, all transitions of the selector outputs
  the empty word and and the output label of any run is empty.
  Applying Lemma~\ref{lem:ergodic} with $ε/(\#A)^k$ and $δ =
  π_q/2$ provides an integer $N_0$ such that for each $n > N_0$, the set
  \begin{displaymath}
    B = \left\{
        u ∈ A^n : \;\;\bigl|\occ{p * u}{q}/n - π_q\bigr| > π_q/2
    \right\}
  \end{displaymath}
  has cardinality at most $ε(\#A)^{n-k}$.  Fix now
  $N = \max(N_0,2k/π_p)$ and let $n$ be such that $n > N$.  If a word~$u$
  of length~$n$ does not belong to~$B$, the run $p * u$ satisfies
  $\occ{p * u}{q} > nπ_q/2 ⩾ k$.  This implies that the length
  of its output label is greater than $k$.  Indeed, the state~$q$ has
  at most $k+1$ occurrences in the run and each transition starting from~$q$
  outputs one symbol.

  Consider the $(\#A)^n$ runs of the form $p * u$ for $u$ of length~$n$.
  Among these runs, at most $ε(\#A)^{n-k}$ many of them do not have
  an output greater than~$k$.  For each $w' \neq w$, $w'$ is the prefix of
  the output label of at most $(\#A)^{n-k}$ many of them.  It follows that
  $w$ is the prefix of the output label of at least
  $(1-ε)(\#A)^{n-k}$ many of them.
\end{proof}

Let $𝒜$ be an automaton with state set $Q$.  We now define and
automaton whose states are the run of length~$n$ in~$𝒜$.  We let
$𝒜^n$ denote the automaton whose state set is $\{ p * w : p ∈
Q, w ∈ A^n\}$ and whose set of transitions is defined by
\begin{displaymath}
  \left\{
    {(p * bw) \textstyle\trans{a} (q * wa)} :
    {p \textstyle\trans{b} q}\text{ in $𝒜$}, \;\; a,b ∈ A
    \text{ and }w ∈ A^{n-1}
  \right\}
\end{displaymath}

The Markov chains associated with the automaton~$𝒜^n$ is called the
\emph{snake} Markov chains.  See Problems 2.2.4, 2.4.6 and 2.5.2 (page~90)
in \cite{Bremaud08} for more details.  It is pure routine to check that the
distribution~$ξ$ of~$𝒜^n$ is given by $ξ_{p * w} = π_p/(\#A)^n$ for each
state~$p$ and each word~$w$ of length~$n$.

\begin{proof}[Proof of theorem~\ref{thm:selection}]
  Let $y$ be the output of the run of~$𝒮$ on~$x$.  By
  Lemma~\ref{lem:recurscc}, the run of~$𝒮$ on~$x$ reaches a recurrent
  SCC.  Therefore it can be assumed without loss of generality
  that the selector~$𝒮$ is strongly connected.

  Let $k$ be a fixed integer.  We claim that for each word~$w$ of length~$k$
  $\lim_{n → ∞} \occ{y[1{:}n]}{w}/n = 1/(\#A)^k$.  With each occurrence of a
  word~$w$ of length~$k$ in~$y$, we associate the occurrence of the
  state~$q$ in the run at which starts the transition that outputs the
  first symbol of~$w$.  Note that transitions starting from~$q$ must be of
  type~I.  Conversely, with each occurrence in the run of such a state, we
  associate the block of length~$k$ of~$y$ starting from that position.

  We fix a state~$p$ such that transitions starting from~$p$ have type~I.
  We first claim that for each integer~$n$ all runs of length~$n$ starting
  from~$p$ have the same frequency in the run.  To prove this claim, we
  apply Corollary~\ref{cor:freqstate} to the automaton~$𝒜^n$ where $𝒜$ is
  the automaton obtained by removing the outputs from~$𝒮$.

  Let $ε>0$ be a positive real number.  By Lemma~\ref{lem:equirun}, there
  is an integer~$n$ such that for each~$w$ on length~$k$, the number of run
  starting from~$p$ outputting $w$ as their first $k$ symbols is between
  $(1-ε)(\#A)^{n-k}$ and~$(\#A)^{n-k}$.  Combining this result with the fact
  that all these runs of length~$n$ have the same frequency, we get that
  the frequency of of each $w$ is between $(1-ε)(\#A)^{-k}$ and~$(\#A)^{k}$.
  Since this is true for each $ε>0$, all words of length~$k$ have the same
  frequency after an occurrence of~$p$.  Since this is true for each
  state~$p$, we get that all words of length~$k$ have the same frequency
  in~$y$.
\end{proof}

\section{Genericity for Markov measures} \label{sec:markov}

In this section we extend Agafonov's result to the more general setting of
shifts of finite type.  In this context, normality is defined through the
Parry measure which is the unique invariant and compatible measure with
maximal entropy.  A sequence is said to be normal if it is generic for that
measure.  We actually prove a slightly stronger result by showing that
genericity for any Markov measure is preserved by finite state selection as
long as the selection is compatible with the measure.  This includes the
case of shifts of finite type because their Parry measure is Markovian.

To obtain such a result, the selection must be perfomed in a compatible way
with the measure and its support.  This boils down to putting some
constraints on the selector to guarantee that if the input sequence is in
the support of the measure, then the output sequence is also in that
support.  Insuring that the output is still in the support is not enough as
it is shown by the following example.  Consider the golden mean shift~$X$
and the selector pictured in Figure~\ref{fig:select1}.  This selector
selects symbols following a~$1$.  If the input sequence~$x$ is in~$X$, the
sequence~$y$ of selected symbols is $0^ℕ = 000⋯$ since $x$ has no
consecutive~$1$s.  Therefore, $y$ is always in~$X$ but genericity is lost.
To prevent this problematic behaviour, the selector is only allowed to
select the next symbol if the last read symbol and the last selected symbol
coincide.  This restriction rules out the previous selector because it does
satisfies this property.

We suppose that a markov measure $μ = μ_{π,P}$ is fixed and we let $X_μ$ be
its support.  We introduce automata and selectors which are compatible with
the shift~$X_μ$.  An automaton~$𝒜$ is compatible with~$X_μ$ if there exists
a function~$ι$ from its state set~$Q$ to~$A$ such that the following
condition is fulfilled.
\begin{itemize} \itemsep0cm
\item[i)] If $p \trans{a} q$ is a transition of~$𝒜$, then $P_{ι(p)a} > 0$
  and $ι(q) = a$.
\end{itemize}
The condition implies that all transitions arriving to a given state~$q$
have the same label~$ι(q)$ and that the label of any path is in the
shift~$X_μ$.  Such an automaton is called $X_μ$-complete if for each pair
$(p,a)$ such that $P_{ι(p)a} > 0$, there exists a transition
$p \trans{a} q$ for some state~$q$.

We continue by defining selectors which are compatible with~$X_μ$.  A
selector~$𝒮$ is compatible with~$X_μ$ exists two functions $ι$ and~$η$ from
its state set~$Q$ to the alphabet~$A$ such that the following two
conditions are fulfilled.
\begin{itemize} \itemsep0cm
\item[i)] If $p \trans{a|a} q$ is a transition of type~I, then $P_{ι(p)a} > 0$,
  $ι(q) = η(q) = a$, and $η(p) = ι(p)$ 
\item[ii)] If $p \trans{a|\emptyword} q$ is a transition of type~II, then
  $P_{ι(p)a} > 0$, $ι(q) = a$ and $η(q) = η(p)$.
\end{itemize}
The condition $P_{ι(p)a} > 0$ states that the selector can only read
consecutive symbols with non-zero transition probability.  The condition
$η(p) = ι(p)$ for the transition $p \trans{a|a} q$ states the last read and
last selected symbols must coincide for the selector to be able to select.
The other conditions states that $ι(q)$ is always the last read symbol, and
that $η(q)$ is the last selected symbol if there is one and that it is equal
to~$η(p)$ otherwise.

\begin{figure}[htbp]
  \begin{center}
  \begin{tikzpicture}[->,>=stealth',initial text=,semithick,auto,inner sep=1.5pt]
    \tikzstyle{every state}=[minimum size=0.4]
    \node[state, initial above] (qE00) at (2,2) {000};
    \node[state] (qE01) at (6,2) {001};
    \node[state] (qE10) at (0,0) {010};
    \node[state] (qE11) at (4,0) {011};
    \node[state] (qO00) at (2,0) {100};
    \node[state] (qO01) at (6,0) {101};
    \node[state] (qO10) at (0,2) {110};
    \node[state] (qO11) at (4,2) {111};
    \path (qE00) edge[bend left=10] node {$0|\emptyword$} (qO00);
    \path (qE00) edge[bend left=10] node {$1|\emptyword$} (qO10);
    \path (qE01) edge[bend left=10] node {$0|\emptyword$} (qO01);
    \path (qE01) edge node[swap] {$1|\emptyword$} (qO11);
    \path (qE10) edge node[swap] {$0|\emptyword$} (qO00);
    \path (qE10) edge[bend left=10,dashed] node {$1|\emptyword$} (qO10);
    \path (qE11) edge[bend left=10] node {$0|\emptyword$} (qO01);
    \path (qE11) edge[bend left=10,dashed] node {$1|\emptyword$} (qO11);
    \path (qO00) edge[bend left=10] node {$0|0$} (qE00);
    \path (qO00) edge node[swap] {$1|1$} (qE11);
    \path (qO01) edge[bend left=10] node {$0|\emptyword$} (qE01);
    \path (qO01) edge[bend left=10] node {$1|\emptyword$} (qE11);
    \path (qO10) edge[bend left=10] node {$0|\emptyword$} (qE00);
    \path (qO10) edge[bend left=10,dashed] node {$1|\emptyword$} (qE10);
    \path (qO11) edge node[swap] {$0|0$} (qE00);
    \path (qO11) edge[bend left=10,dashed] node {$1|1$} (qE11);
  \end{tikzpicture}
  \end{center}
  \caption{A selector compatible with the golden mean shift}
  \label{fig:select2}
\end{figure}
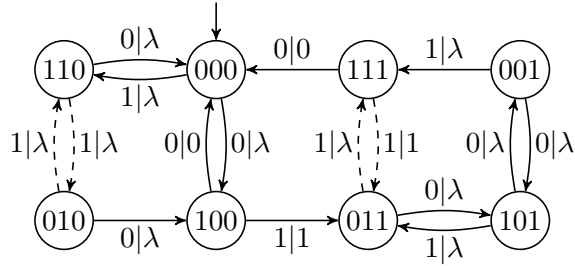

The selector pictured in Figure~\ref{fig:select2} is compatible with the
golden mean shift.  It selects symbols at even positions (starting
from~$1$) if it is possible, that is, if the last read symbol and the last
selected symbol coincide.  The dashed edges are useless if the input
sequence is in the golden mean shift.  In that case, the output sequence is
also in the golden mean shift.  Each state is labelled by~$prs$ where $p ∈
\{0, 1\}$ is the parity of the number of read symbols so far, $r ∈ \{0,
1\}$ is the last read symbol and $s ∈ \{0, 1\}$ the last selected symbol.
The two functions $ι$ and~$η$ can be defined by $ι(prs) = r$ and $η(prs) =
s$.

The following theorem states that selection with compatible selectors
preserves genericity for Markov measures.  The input sequence~$x$ must be
assumed to be in the shift~$X_μ$ because compatible selectors only read
sequences from~$X_μ$.
\begin{theorem} \label{thm:markov-selection}
  Let $μ$ be a Markov measure and let $x$ be a sequence in~$X_μ$ which is
  $μ$-generic.  For each oblivious selector~$𝒮$ compatible with~$X_μ$, the
  output~$𝒮(x)$ of~$𝒮$ on~$x$ belongs to~$X_μ$ and is $μ$-generic.
\end{theorem}

The previous theorem can be applied to the Parry measure~$μ$ of a shift~$X$
of finite type because the suppport of~$μ$ is actually $X_μ = X$.

We start with the definition of the conditional measures induced by~$μ$.
For each symbol $a ∈ A$, we let $μ_a$ denote the \emph{conditional measure}
defined by
\begin{displaymath}
  μ_a(a_1 a_2 ⋯ a_n) = P_{aa_1} P_{a_1a_2} ⋯ P_{a_{n-1} a_n}.
\end{displaymath}
Note that the measures~$μ_a$ might not be invariant.  Since $π$ is the
stationnary distribution, the measure~$μ$ can be recovered from the
measures~$μ_a$ by the formula $μ = \sum_{a ∈ A} π_aμ_a$.

The following lemma shows that the set of runs starting from a state~$p$
and outputting a fixed word~$w$ is not too large.  This is the analog of
Lemma~\ref{lem:upper} in the context of Markov measures.
\begin{lemma} \label{lem:markov-upper}
  Let $𝒮$ be an oblivious selector compatible with~$\mu$.  For each
  state~$p$ of~$𝒮$, each integer~$n ⩾ 0$, and word~$w ∈ A^*$ such that
  $|w| ⩽ n$, then the inequality
  $μ_{ι(p)}(\{ u ∈ A^n : p \trans{u|v} q \text{ and } w ⊑ v\}) <
  μ_{η(p)}(w)$ holds.
\end{lemma}
\begin{proof}
  Let $U$ be the set
  $\{ u ∈ A^n : p \trans{u|v} q \text{ and } w ⊑ v\}$.  
  The proof is carried out by induction on the integer~$n$.  If $n = 0$,
  the set $U$ is $U = \{ \emptyword \}$ and $w$ must be the empty
  word~$\emptyword$. The inequality is then satisfied because both measures
  are equal to~$1$.  We now suppose that $n ⩾ 1$. Since the selector is
  oblivious, all transitions starting from state~$p$ have the same type,
  either type~I or type~II.  We distinguish two cases depending on the
  type of these transitions.

  We first suppose that all transitions starting from state~$p$ have
  type~I.  Let us write $w = aw'$ where $a$ is a symbol and $w'$ a word.
  Consider the transition $p \trans{a|a} p'$.  The compatibility of~$𝒮$
  with~$μ$ implies that $ι(p) = η(p)$ and $ι(p') = η(q) = a$.  All runs
  starting from~$p$ such that $w$ is a prefix of the output label must use
  this transition as a first transition.  Applying the induction hypothesis
  to $p'$, $n-1$ and~$w'$ gives that $μ_a(U') < μ_a(w')$ where
  $U' = \{ u ∈ A^{n-1} : p' \trans{u|v} q \text{ and } w ⊑ v\}$.  Since
  $U = aU'$, the result follows from the equalities
  $μ_{ι(p)}(U) = P_{ι(p)a}μ_a(U')$ and $μ_{η(p)}(w) = P_{η(p)a}μ_a(w')$.

  We now suppose that all transitions starting from state~$p$ have type~II,
  that is, have the form $p \trans{a|\emptyword} p_a$ for each symbol~$a$.
  The compatibility of~$𝒮$ with~$μ$ implies that $ι(p_a) = a$ and
  $η(p_a) = η(p)$ for each $a ∈ A$.  All runs of length~$n$ starting
  from~$p$ have an output label of length at most $n-1$.  Therefore, if
  $|w| = n$, there is no run such that $w$ is prefix of its output label
  and the inequality is trivially satisfied.  If $|w| ⩽ n-1$, applying the
  induction hypothesis to each~$p_a$, $n-1$ and~$w$, gives that
  $μ_a(U_a) < μ_{η(p)}(w)$ where
  $U_a = \{ u ∈ A^{n-1} : p_a \trans{u|v} q \text{ and } w ⊑ v\}$.  Since
  $U = \bigcup_{a ∈ A} aU_a$, the result follows from the equalities
  $μ_{ι(p)}(U) = \sum_{a ∈ A} P_{ι(p)a}μ_a(U_a)$ and
  $μ_{η(p)}(w) = \sum_{a ∈ A} P_{ι(p)a}μ_{η(p_a)}(w) = μ_{η(p_a)}(w)$.
\end{proof}

Some of the bounds are again obtained using the ergodic theorem for Markov
chains \cite[Thm~4.1]{Bremaud08}.  For that purpose, we associate a Markov
chain $ℳ$ to each strongly connected automaton~$𝒜$ which is compatible
with~$X_μ$  and $X_μ$-complete.  This means that there is a function~$ι$
from~$Q$ to~$A$ such that if $p \trans{a} q$ is a transition, then $ι(q) =
a$.   For simplicity, we assume that the state set~$Q$
of~$𝒜$ is the set $\{1, …,\#Q\}$.

The state set of the Markov chain is the same set $\{1, …,\#Q\}$.  The
transition matrix of the Markov chain is the matrix
$\hat{P} = (\hat{P}_{pq})_{1 ⩽ p,q ⩽ \#Q}$ where each entry~$\hat{P}_{pq}$
is equal to $P_{ι(p)a} = P_{ι(p)ι(q)}$ if $p \trans{a} q$ is a transition
of~$𝒜$ and $0$ otherwise.  Since the automaton is assumed to be
deterministic and $X_μ$-complete, the matrix~$\hat{P}$ is stochastic.  If the
automaton~$𝒜$ is strongly connected, the Markov chain is irreducible and it
has therefore a unique stationary distribution~$\hat{π}$ such that
$\hat{π} \hat{P} = \hat{π}$.  The vector~$\hat{π}$ is called the
\emph{distribution} of~$𝒜$.  The matrix~$\hat{P}$ and its stationary
distribution~$\hat{π}$ define a Markov measure
$\hat{μ} = μ_{\hat{π},\hat{P}}$ on finite runs of~$𝒜$.  The link between
the measures $μ$ and~$\hat{μ}$ is that
$\hat{μ}(p * u) = \hat{π}_p μ_{ι(p)}(u)$ for each state~$p$ and each
word~$u$.

\begin{lemma} \label{lem:markov-ergodic}
  Let $𝒜$ be a strongly connected deterministic and complete
  automaton and let $π$ be its distribution.  For each real numbers
  $ε,δ > 0$, there exists an integer~$N$ such that for each
  integer $n > N$
  \begin{displaymath}
   μ\left(\left\{
        u ∈ A^n :
        ∃ p,q ∈ Q \;\;\bigl|\occ{p * u}{q}/n - \hat{π}_q\bigr| > δ
    \right\}\right) < ε
  \end{displaymath}
\end{lemma}
The lemma is stated for the measure~$μ$ but the ergodic theorem is valid
for any initial distribution.  The result is therefore also valid for the
conditional measures~$μ_a$.
\begin{proof}
  The proof is a mere application of the ergodic theorem for Markov chains
  \cite[Thm~4.1]{Bremaud08}.
\end{proof}

\begin{corollary} \label{cor:markov-freqstate}
  Let $𝒜$ be a deterministic and strongly connected automaton
  and let $π$ its distribution.  Let $ρ$ be the run of~$𝒜$
  on a $μ$-generic sequence~$x$.  Then for each state~$q$
  \begin{displaymath}
    \lim_{n → ∞}{\frac{\occ{ρ[1{:}n]}{q}}{n}} = \hat{π}_q.
  \end{displaymath}
  where $ρ[1{:}n]$ is the finite run made of the first $n$ transitions
  of~$ρ$
\end{corollary}
\begin{proof}
  Since $\sum_{q \in Q}{\hat{π}_q} = 1$, it suffices to prove that
  $\liminf_{n → ∞}{\occ{ρ[1{:}n]}{q}/n}  ⩾ \hat{π}_q$ holds for each
  state~$q$.

  Let $ε>0$ be a positive real number.  Applying
  Lemma~\ref{lem:markov-ergodic} with $δ = ε$ provides an
  integer~$k$ such that
  \begin{displaymath}
    μ\left(\left\{
        u ∈ A^k :
        ∃ p \;\;\bigl|\occ{p * u}{q}/k - \hat{π}_q\bigr| > ε
    \right\}\right) < ε.
  \end{displaymath}
  The run~$ρ$ is then factorized
  \begin{displaymath}
    ρ = p_0 \trans{u_0} p_1 \trans{u_1} p_2  \trans{u_2} p_3 ⋯
         = (p_0 * u_0)(p_1 * u_1)(p_2 * u_2)⋯
  \end{displaymath}
  where each word~$u_i$ is of length~$k$ and $x = u_0u_1u_2⋯$.
  Since $x$ is $μ$-generic, there is an integer~$N$ such that for each
  $n > N$ the cardinality of the set $\{ i < n : u_i = u \}$ is greater
  than $(1-ε)nμ(u)$ for each word $u ∈ A^k$.

  \begin{align*}
    \liminf_{n → ∞}{\frac{\occ{ρ[1{:}n]}{q}}{n}} & =
           \liminf_{n → ∞}{\frac{\occ{ρ[1{:}nk]}{q}}{nk}} \\
    & = \frac{1}{nk}\sum_{i = 0}^{n-1}{\occ{p_i * u_i}{q}} \\
    & ⩾ \frac{1}{nk}\sum_{u ∈ A^k}{\#\{ i < n : u_i = u \}
        × \min_{p ∈ Q}\occ{p * u}{q}} \\
    & ⩾ \frac{1}{nk}\sum_{u ∈ A^k \setminus B}
         {((1-ε)nμ(u))(k(\hat{π}_q-ε))} \\
    & = (1-ε)^2(\hat{π}_q-ε)
  \end{align*}
  Since this inequality holds for each real number $ε>0$, we have
  proved that $\liminf_{n → ∞}{\occ{ρ[1{:}n]}{q}/n} ⩾ \hat{π}_q$.
\end{proof}

\begin{lemma} \label{lem:markov-recurscc}
  Let $𝒜$ be an automaton compatible with~$μ$ and let $ρ$ be a run in~$𝒜$
  on a $μ$-generic sequence in~$X_μ$.  The run~$ρ$ reaches a recurrent
  strongly connected component of~$𝒜$.
\end{lemma}
\begin{proof}
  We claim that for each SCC~$C$ which is not recurrent, there exists a
  word~$w$ with $μ(w) > 0$ and starting with a symbol~$a$ such that from
  any state~$q$ in~$C$ such that $P_{ι(q)a}> 0$ the run $q * w$ leaves~$C$.

  We fix a symbol~$a$.  Let $\{q_1, … ,q_n\}$ be the set of states~$q$
  in~$C$ such $P_{ι(q)a} > 0$.  We construct a sequence $w_0,w_1,…,w_n$ of
  words such that if $i ⩽ j$, then the run $q_i* w_j$ leaves $C$.  We set
  $w_0 = \emptyword$ and the statement is true.  Suppose that $w_0,…, w_k$
  have been already chosen and consider the state $p_k = q_{k+1} ⋅ w_k$.
  If this state~$p_k$ is already out of~$C$, we set $w_{k+1} = w_k$.
  Otherwise, since $C$ is not recurrent, there is a word~$v_k$ such that
  $p_k ⋅ v_k$ is out of~$C$: we set $w_{k+1} = w_kv_k$ so that
  $q_{k+1}⋅ w_{k+1} = p_k⋅ v_k$ is out of~$C$.

  The run~$ρ$ reaches a last SCC~$C$.  Suppose by constriction that $C$ is
  not recurrent.  By the previous claim there is a word $w = aw'$ such that
  $μ(w) > 0$ and such that for any state~$q$ in~$C$ with $P_{ι(q)a} > 0$,
  $q * w$ leaves $C$.  Since $x$ is $μ$-generic, the word~$w$ occurs
  infinitely often in~$x$.  Let $q$ be state of~$C$ reached by the run~$ρ$
  before an occurrence of~$w$.  Since $x$ is in~$X_μ$, $P_{ι(q)a} > 0$.
  This is a contradiction because $q * w$ leaves $C$ while $C$ is supposed
  to be the last SCC reached by~$ρ$.
\end{proof}

\begin{lemma} \label{lem:markov-equirun}
  Let $𝒮$ be a strongly connected selector.  For each integer~$k$
  and each real number $ε > 0$, there exists an integer~$N$ such
  that for each integer~$n > N$, each state~$p$ and each word~$w$ of
  length~$k$, the inequalities
  $(1-ε)μ_{η(p)}(w) < μ_{ι(p)}(\{ u ∈ A^n : p \trans{u|v} q \text{ and } w
  \sqsubseteq v\}) < μ_{η(p)}(w)$ hold.
\end{lemma}
\begin{proof}
  Let $p$ be any state. The upper bound~$μ_{η(p)}(w)$ has been already
  proved in Lemma~\ref{lem:markov-upper}.  It remains to prove the lower
  bound.

  Let fix a state~$q$ such that the transitions starting from~$q$ are of
  type~I.  If no such state exists, all transitions of the selector outputs
  the empty word and the output label of any run is empty.
  Applying Lemma~\ref{lem:markov-ergodic} with $εμ_{η(p)}(w)$ and $δ =
  π_q/2$ provides an integer $N_0$ such that for each $n > N_0$,
  \begin{displaymath}
    μ_{ι(p)}\left(\left\{
        u ∈ A^n : \;\;\bigl|\occ{p * u}{q}/n - π_q\bigr| > π_q/2
    \right\}\right) < εμ_{η(p)}(w).
  \end{displaymath}
  Fix now $N = \max(N_0,2k/π_p)$ and let $n$ be such that $n > N$.  If a
  word~$u$ of length~$n$ does not belong to the small set above, the run
  $p * u$ satisfies $\occ{p * u}{q} > nπ_q/2 ⩾ k$ for each state~$p$.
  This implies that the length of its output label is greater than $k$.
  Indeed, the state~$q$ has at most $k+1$ occurrences in the run and each
  transition starting from~$q$ outputs one symbol.

  Consider the $(\#A)^n$ runs of the form $p * u$ for $u$ of length~$n$.  The
  measure of those having an output smaller than~$k$ is less than
  $εμ_{η(p)}(w)$.  For each $w' \neq w$, the measure of those having $w'$
  as prefix of length~$k$ of their output label is at most $μ_{η(p)}(w)$.
  It follows that the measure of those having $w$ as prefix of length~$k$
  of their output label is at most $(1-ε)μ_{η(p)}(w)$.
\end{proof}

Let $𝒜$ be an automaton with state set~$Q$.  We now define an automaton
whose states are the run of length~$n$ in~$𝒜$.  We let $𝒜^n$ denote the
automaton whose state set is $\{ p * u : p ∈ Q, u ∈ A^n\}$ and whose set of
transitions is defined by
\begin{displaymath}
  \left\{
    {(p * bu) \textstyle\trans{a} (q * ua)} :
    {p \textstyle\trans{b} q}\text{ in $𝒜$}, \;\; a,b ∈ A
    \text{ and }u ∈ A^{n-1}
  \right\}
\end{displaymath}

The Markov chains associated with the automaton~$𝒜^n$ is called the
\emph{snake} Markov chains.  See Exercises 2.2.4, 2.4.6 and 2.5.2 in
\cite{Bremaud08} for more details.  It is pure routine to check that the
distribution~$\hat{ξ}$ of~$𝒜^n$ is given by
$\hat{ξ}_{p * w} = \hat{π}_pμ_{η(p)}(w)$ for each state~$p$ and each
word~$w$ of length~$n$.

\begin{proof}[Proof of theorem~\ref{thm:markov-selection}]
  Let $y$ be the output of the run of~$𝒮$ on~$x$.  By
  Lemma~\ref{lem:markov-recurscc}, the run of~$𝒮$ on~$x$ reaches a
  recurrent strongly connected component.  Therefore it can be assumed
  without loss of generality that the selector~$𝒮$ is strongly connected.

  Let $k$ be a fixed integer.  We claim that for each word~$w$ of length~$k$
  $\lim_{n → ∞} \occ{y[1{:}n]}{w}/n = μ(w)$.  With each occurrence of a
  word~$w$ of length~$k$ in~$y$, we associate the occurrence of the
  state~$q$ in the run from which starts the transition that outputs the
  first symbol of~$w$.  Note that transitions starting from~$q$ must be of
  type~I.  Conversely, with each occurrence in the run of such a state, we
  associate the block of length~$k$ of~$y$ starting from that position.

  We fix a state~$p$ such that transitions starting from~$p$ have type~I.
  We first claim that for each integer~$n$, each run $p * u$ of length~$n$
  starting from~$p$ has a frequency of $μ_{η(p)}(w)$. To prove this claim,
  we apply Corollary~\ref{cor:markov-freqstate} to the automaton~$𝒜^n$
  where $𝒜$ is the automaton obtained by removing the outputs from~$𝒮$.

  Let $ε > 0$ be a positive real number.  By
  Lemma~\ref{lem:markov-equirun}, there is an integer~$n$ such that for
  each~$w$ on length~$k$, the measure~$μ_{ι(p)}$ of all runs starting
  from~$p$ outputting $w$ as their first $k$ symbols is between
  $(1-ε)μ_{η(p)}(w)$ and~$μ_{η(p)}(w)$.  Combining this result with the
  fact that each run $p * u$ of length~$n$ occurs after state~$p$ with a
  frequency equal to~$μ_{ι(p)}(u)$, we get that the frequency of each
  word~$w$ is between $(1-ε)μ_{η(p)}(w)$ and~$μ_{η(p)}(w)$.  Since this is
  true for each $ε > 0$, all words of length~$k$ have a frequency after
  state~$p$ equal to its measure $μ_{η(p)}(w)$.  Since this is true for
  each state~$p$, we get that each word of length~$k$ have a frequency
  equal to~$μ(w)$ in~$y$.
\end{proof}

\section*{Conclusion}

As a conclusion, we would like to mention a few extensions of our results.
Agafanov's theorem deals with prefix selection: a given digit is selected
if the prefix of the word up to that digit belongs to a fixed set of finite
words.  \emph{Suffix selection} is defined similarly: a given digit is
selected if the suffix of the word from that digit belongs to a fixed set
of sequences.  It has been shown in~\cite{BecherCartonHeiber15} that suffix
selection also preserves normality as long as the fixed set of sequences is
regular.  Let us recall that a set of sequences is regular if it can be
accepted by non-deterministic B\"uchi or by a deterministic Muller
automaton \cite{PerrinPin04}.  The proof given
in~\cite{BecherCartonHeiber15} is based on the characterization of
normality by non-compressibility.  The proof techniques developed here to
prove Agafanov's theorem can be adapted to also prove directly the result
about suffix selection.

The prefix and suffix selections considered so far are usually called
\emph{oblivious} because the digit to be selected is not included to either
the prefix or the suffix taken into account.  Non-oblivious does not
preserve in general normality but it does for a restricted class of sets of
finite words called group languages \cite{CartonVandehey20}.  Group
languages are sets of words which are accepted by deterministic automata
such that each symbol induces a permutation of the states.  This later
property means that for each symbol~$a$, the function which maps each
state~$p$ to the state~$q$ such that $p \trans{a} q$ is a permutation of
the state set.  The techniques presented in this paper can also be adapted
to prove such a result.

\bibliographystyle{plain}
\bibliography{selection}

\begin{thebibliography}{10}

\bibitem{Agafonov68}
V.~N. Agafonov.
\newblock Normal sequences and finite automata.
\newblock {\em Soviet Mathematics Doklady}, 9:324--325, 1968.

\bibitem{AlvarezCarton19}
N.~{\'A}lvarez and O.~Carton.
\newblock On normality in shifts of finite type.
\newblock {\em CoRR}, abs/1807.07208, 2018.

\bibitem{BecherCartonHeiber15}
V.~Becher, O.~Carton, and P.~A. Heiber.
\newblock Normality and automata.
\newblock {\em Journal of Computer and System Sciences}, 81(8):1592--1613,
  2015.

\bibitem{BecherHeiber13}
V.~Becher and P.~A. Heiber.
\newblock Normal numbers and finite automata.
\newblock {\em Theoretical Computer Science}, 477:109--116, 2013.

\bibitem{Borel09}
\'E. Borel.
\newblock Les probabilit\'{e}s d\'{e}nombrables et leurs applications
  arithm\'{e}tiques.
\newblock {\em Rendiconti del Circolo Matematico di Palermo}, 27:247--271,
  1909.

\bibitem{Bremaud08}
P.~Br{\'e}maud.
\newblock {\em Markov Chains: Gibbs Fields, Monte Carlo Simulation, and
  Queues}.
\newblock Springer, 2008.

\bibitem{BroglioLiardet92}
A.~Broglio and P.~Liardet.
\newblock Predictions with automata. {S}ymbolic dynamics and its applications.
\newblock {\em Contemporary Mathematics}, 135:111--124, 1992.
\newblock Also in Proceedings AMS Conference in honor of R. L. Adler. New Haven
  CT - USA 1991.

\bibitem{CartonVandehey20}
O.~Carton and J.~Vandehey.
\newblock Preservation of normality by non-oblivious group selection.
\newblock {\em ToCS}, 2020.

\bibitem{Dai04}
J.~Dai, J.~Lathrop, J.~Lutz, and E.~Mayordomo.
\newblock Finite-state dimension.
\newblock {\em Theoretical Computer Science}, 310:1--33, 2004.

\bibitem{Kitchens98}
B.~P. Kitchens.
\newblock {\em Symbolic Dynamics}.
\newblock Springer, 1998.

\bibitem{LindMarcus92}
D.~Lind and B.~Marcus.
\newblock {\em An Introduction to Symbolic Dynamics and Coding}.
\newblock Cambridge University Press, 1992.

\bibitem{Madritsch18}
M.~Madritsch.
\newblock Normal numbers and symbolic dynamics.
\newblock In {\em Sequences}, chapter~8. Cambridge University Press, 2018.

\bibitem{OConnor88}
M.~G. O'Connor.
\newblock An unpredictability approach to finite-state randomness.
\newblock {\em Journal of Computer and System Sciences}, 37(3):324--336, 1988.

\bibitem{PerrinPin04}
D.~Perrin and J.-{\'E}. Pin.
\newblock {\em Infinite Words}.
\newblock Elsevier, 2004.

\bibitem{Sakarovitch09}
J.~Sakarovitch.
\newblock {\em Elements of automata theory}.
\newblock Cambridge University Press, 2009.

\bibitem{SchnorrStimm71}
C.~P. Schnorr and H.~Stimm.
\newblock Endliche automaten und zufallsfolgen.
\newblock {\em Acta Informatica}, 1:345--359, 1972.

\end{thebibliography}
\end{document}